\documentclass[twoside,bezier,12pt, reqno]{amsart}
\usepackage{amssymb,latexsym,epsfig,mathrsfs,graphpap,graphics,amsmath,amssymb}
\usepackage{amstext}
\usepackage{amsthm}
\usepackage[utf8]{inputenc}
\usepackage[english]{babel}
\usepackage{helvet}
\usepackage{courier}

\newtheorem{theorem}{Theorem}[section]
\newtheorem{lemma}{Lemma}[section]

\theoremstyle{Definition}
\newtheorem{definition}{Definition}[section]

\theoremstyle{remark}

\numberwithin{equation}{section}

\begin{document}

\begin{flushleft}
 {\bf\Large { Uncertainty Inequalities for 3D Octonionic-valued Signals Associated with Octonion Offset Linear Canonical Transform}}

\parindent=0mm \vspace{.2in}

{\bf{M. Younus Bhat$^{1},$ and Aamir H. Dar$^{2}$ }}
\end{flushleft}

{{\it $^{1}$ Department of  Mathematical Sciences,  Islamic University of Science and Technology Awantipora, Pulwama, Jammu and Kashmir 192122, India.E-mail: $\text{g gyounusg@gmail.com}$}}

{{\it $^{2}$ Department of  Mathematical Sciences,  Islamic University of Science and Technology Awantipora, Pulwama, Jammu and Kashmir 192122, India.E-mail: $\text{ahdkul740@gmail.com}$}}

\begin{quotation}
\noindent
{\footnotesize {\sc Abstract.} The octonion offset linear canonical transform ($\mathbb O-$OLCT) can be defined as a  time-shifted and frequency-modulated version of the octonion linear canonical
transform ($\mathbb O-$LCT),  a
more general framework of most existing  signal
processing tools. In this  paper, we first define the ($\mathbb O-$OLCT) and provide its closed-form representation. Based on this fact, we study some fundamental properties of proposed transform including inversion formula, norm split and energy conservation. The crux of the paper lies in the generalization of  several well known  uncertainty relations for the ($\mathbb O-$OLCT)  that  include Pitt’s inequality, logarithmic uncertainty inequality, Hausdorff-Young inequality
 and local uncertainty inequalities.\\


}
\end{quotation}
\section{ \bf Introduction}
\noindent
The hyper-complex Fourier transform(FT)  is of the interest in the present era. It treats multi-channel signals as an algebraic whole without losing the spectral relations. Presently, many hyper-complex FTs exists in literature which are defined by different approaches, see \cite{5a,6a}. The developing interest in hyper-complex FTs including applications in  watermarking, color image processing, image filtering,  pattern recognition and edge detection \cite{1b}-\cite{6b}. Among the various hyper-complex FTs, the most basic ones are  the quaternion Fourier transforms(QFTs). QFTs are most widely studied in recent years because of its wide applications in optis and signal processing. QFT\cite{8b} is very useful in Cayley-Dickson algebra of order 4(Quaternions) as it is a substitute to the two-dimensional complex Fourier transform (CFT). Various properties and applications of the QFT were established in \cite{9a}-\cite{12a}. The generalization of quaternion Fourier transform (QFT)is quaternion linear canonical transform (QLCT), which is more effective signal processing tool than QFT due to its extra parameters, see\cite{13a,14a,15a,16a,17a,18a,WLCT}.\\
Later, the quaternion linear canonical transform (QLCT) with four parameters  has been generalized to a six parameter transform  known as quaternion offset linear canonical transform (QOLCT). Due to the time shifting  and frequency modulation parameters, the QOLCT has gained more flexibility over classical QLCT.  Hence has found wide applications in image and signal processing, see \cite{own1,own2,own3,gen}.\\
On the other hand  the Cayley-Dickson algebra of order 8 is known as octonion algebra  which deserve special attention  in the  hyper-complex signal processing. The   octonion Fourier transform ($\mathbb O-$FT) was proposed by Hahn and Snopek in 2011\cite{hn}. From then $\mathbb O-$FT is becoming the hot area of research in modern signal processing community. Some properties and uncertainty relations and applications associated with $\mathbb O-$FT have been studied, see\cite{w1,w2,w3,w4}.In 2021 Gao and Li \cite{li} proposed octonion linear canonical transform($\mathbb O-$LCT) as a generalization of  $\mathbb O-$FT by substituting the Fourier kernel with the LCT kernel. They established some vital properties like inversion formula, isometry, Riemann-Lebesgue lemma and proved Heisenberg's and Donoho-Stark's  uncertainty principles. The generalization of $\mathbb O-$FT to other transforms viz linear canonical transform, offset linear canonical transform, Stockwell transform, Quadratic phase transform etc. is still in its infancy.

So motivated and inspired  by this, we shall propose the novel octonion offset linear canonical transform ($\mathbb O-$OLCT)
 which have never been proposed up to date, therefore it is worthwhile to rigorously study the Octonion offset linear canonical transform $\mathbb O-$OLCT
 which can be productive for signal processing theory and applications.\\
 The highlights of the paper are pointed out below:\\
 \begin{itemize}
 \item  To introduce a novel integral transform coined as the octonion offset linear canonical transform ($\mathbb O-$OLCT).\\
 \item To study the fundamental properties of the proposed transform, including the closed-form representation, norm split into quaternions, inversion formula and energy conservation.\\
     \item To formulate several classes of uncertainty inequalities, such as the Pitt’s inequality, logarithmic uncertainty inequality, Hausdorff-Young inequality
 and local uncertainty inequality associated with the octonion offset linear canonical transform ($\mathbb O-$OLCT).\\
 \end{itemize}
      The rest of the paper is organized as follows: In Section 2, some general definitions and basic properties of octonions  are summarized. The definition and the properties of the $\mathbb O-$OLCT are studied in Section 3. In Section 3,
we develop a series of uncertainty inequalities such as the Pitt’s inequality, logarithmic uncertainty inequality, Hausdorff-Young inequality
 and local uncertainty inequality   associated with the $\mathbb O-$OLCT. Finally, a conclusion is extracted in Section 5.
\section{\bf Preliminaries}
\label{sec 2}
In this section,  we collect some basic facts on the octonion algebra and the offset linear canonical transform(OLCT), which will be needed throughout the paper.

\subsection{\bf Octonion algebra} \  \\
The octonion algebra denoted by $\mathbb O,$ \cite{23} is generated by the eighth-order Cayley-Dickson construction. According to  His
construction, a hypercomplex number  $o\in\mathbb O$ is an ordered pair of quaternions $q_0,q_1\in\mathbb H$
\begin{eqnarray}\label{o1}
\nonumber o&=&(q_0,q_1)\\
\nonumber&=&((z_0,z_1),(z_2,z_3))\\
\nonumber&=&q_0+q_1.e_4\\
\nonumber&=&(z_0+z_1.e_2)+(z_2+z_3.e_2).e_4\\
\end{eqnarray}
which has equivalent form
\begin{equation}\label{o2}
o=s_o+\sum_{i=1}^7s_ie_i=s_0+s_1e_1+s_2e_2+s_3e_3+s_4e_4+s_5e_5+s_6e_6+s_7e_7
\end{equation}
that is $o$ is a hypercomplex number defined by eight real numbers $s_i,i=0,1,\dots,7$ and seven imaginary units $e_i$ where $i=1,2,\dots,7.$ The octonion algebra is non-commutative and non-associative algebra. The multiplication
of imaginary units in the Cayley-Dickson algebra of octonions are presented in Table I or in diagram called Fano scheme, shown in Figure 1.\\
\begin{center} Table I \end{center}
\begin{center}{\small Multiplication Rules in Octonion Algebra.} \end{center}
\begin{equation*}\label{table}
\begin{array}{|c|c|c|c|c|c|c|c|c|}
\hline
  \cdot & 1 & e_1 & e_2 & e_3 & e_4 & e_5 & e_6 & e_7 \\
\hline
  1 & 1 & e_1 & e_2 & e_3 & e_4 & e_5 & e_6 & e_7 \\
\hline
  e_1 & e_1 & -1 & e_3 & -e_2 & e_5 & -e_4& -e_7& e_6 \\
\hline
  e_2 & e_2 & -e_3 & -1 & e_1 & e_6 & e_7 & -e_4 & -e_5 \\
\hline
  e_3 & e_3 & e_2 & -e_1 & -1 & e_7 & -e_6 & e_5 & -e_4 \\
\hline
  e_4 & e_4 & -e_5 & -e_6 & -e_7 & -1 & e_1 & e_2 & e_3 \\
  \hline
  e_5 & e_5 & e_4 & -e_7 & e_6 & -e_1 & -1 & -e_3 & e_2\\
\hline
  e_6 & e_6 & e_7 & e_4 & -e_5 & -e_2 & e_3 & -1 & -e_1 \\
  \hline
  e_7 & e_7 & -e_6 & e_5& e_4 & -e_3 & -e_2 & e_1 & -1 \\
  \hline

\end{array}\end{equation*}
The conjugate of an octonion is defined as
\begin{equation}\label{o3}
\overline{o}=s_0-s_1e_1-s_2e_2-s_3e_3-s_4e_4-s_5e_5-s_6e_6-s_7e_7
 \end{equation}
 Therefore norm is defined by $|o|=\sqrt{o\overline{o}}$ and $|o|^2=\sum_{i=o}^7s_i.$ Also $|o_1o_2|=|o_1||o_2|,\forall o_1,o_2\in \mathbb O.$

From (\ref{o1}) it is evident that every $o\in \mathbb O$ can be represented in quaternion form as
\begin{equation}\label{o4}
o=a+be_4
\end{equation}
where $a=s_0+s_1e_1+s_2e_2+s_3e_3$ and $b=s_4+s_5e_1+s_6e_2+s_7e_3$ are both quaternions. By direct verification we have following lemma.
\begin{lemma}\label{lem1}Let $a,b\in \mathbb H,$ then\\
(1)\quad $e_4a=\overline ae_4;$ \qquad\quad(2)\quad$e_4(ae_4)=-\overline a;$\qquad\quad$(3)\quad(ae_4)e_4=-a;$\\
(4)\quad$a(be_4)=(ba)e_4$;\quad\quad(5)\quad$(ae_4)b=(a\overline b)e_4;$\quad\quad(6)\quad$(ae_4)(be_4)=-\overline b a.$
\end{lemma}
It is clear from above Lemma that, for an octonion $a+be_4, a,b \in \mathbb H,$ we have \\
\begin{equation}\label{o5} \overline{a+be_4}=\overline{a}-be_4\end{equation}
and \begin{equation}\label{o6} |a+be_4|^2=|a|^2+|b|^2.\end{equation}

An octonion-valued function $f:\mathbb R^3\longrightarrow\mathbb O$ has following explicit form
\begin{eqnarray}\label{ofun}
\nonumber f(x)&=&f_0+f_1(x)e_1+f_2(x)e_2+f_3(x)e_3+f_4(x)e_4+f_5(x)e_5+f_6(x)e_6+f_7(x)e_7\\
\nonumber&=&g(x)+h(x)e_4\\
\end{eqnarray}

where each $f_i(x)$ is a real vaiued functions, $g,h \in \mathbb H$ and $x=(x_1,x_2,x_3)\in \mathbb R^3.$
For  each octonion-valued  function $f(x)$ over $\mathbb R^3$ and $1\le p<\infty,$ the $L^p-$norm of $f$ is defined by
\begin{equation}\label{lpnorm}
\|f\|^p_p=\int_{\mathbb R^3}|f(x)|^pdx.
\end{equation}
And for $p=\infty$, then the $L^\infty$-norm is defined by
\begin{equation}\label{linfty}
\|f\|_\infty=esssup_{x\in\mathbb R^3}|f(x)|.
\end{equation}

\subsection{Offset linear canonical transform}\ \\
The  offset linear canonical transform(OLCT)\cite{23m} of any function $f:\mathbb R\longrightarrow \mathbb O$ with respect to the  matrix parameter $A=(a,b,c,d,e,\tau,\eta)$ is defined as

\begin{equation}\label{dOLCT}
\mathcal O_{A}\big[f(x)\big] w)= \int_{\mathbb R} f(x)K_{A}^i(x,w)dx.
\end{equation}

with
\begin{equation}\label{eqnk}
K_{A}(x,w)=\dfrac{1}{\sqrt{2\pi |b|}} e^{\frac{{i}}{2b}\big[ax^2-2x(w-\tau)-2w(d\tau-b\eta)+d(w^2+\tau^2)-\frac{\pi}{2}\big]},b\neq0\,
\end{equation}
where $e^{\frac{-i\pi}{4}}$ is the  polar form of $\frac{1}{\sqrt i}.$\\
And the Offset linear canonical transform in quaternion\cite{m18} setting is given as:\\

 Let
	$A_s=\begin{bmatrix}
	a_s&b_s|&\tau_s\\
	c_s&d_s|&\eta_s
	\end{bmatrix}\in \mathbb{R}^{2\times 2}$ be a matrix parameter satisfying ${\rm det}(A_s)=1$, for $s=1,2$. Then  QLCT of signal $f\in L^2\left( \mathbb{R}^2,\mathbb{H}\right)$ is defined by
	\begin{equation}
	\label{dQLCT}
	\mathcal{O}_{i,j}^{A_1,A_2}[f]({w})=
\int_{\mathbb{R}^2}f({x})K_{A_1}^{{i}}(x_1,w_1)K_{A_2}^{{j}}(x_2,w_2)\rm{d}{x},
	\end{equation}	
	where ${w}=(w_1,w_2),x=(x_1,x_2)\in \mathbb{R}^{2}$ , and  the kernel signals $K_{A_1}^{{i}}(x_1,w_1)$, $K_{A_2}^{{j}}(x_2,w_2)$ are respectively given by
	\begin{equation}\label{eqn k1}
K_{A_1}^{{i}}(x_1,w_1)=\dfrac{1}{\sqrt{2\pi b_1{i}}} \,e^{\frac{{i}}{2b_1}\big[a_1x_1^2-2x_1(w_1-\tau_1)-2w_1(d_1\tau_1-b_1\eta_1)+d_1(w_1^2+\tau_1^2)\big]},b_1\neq0\,
\end{equation}

\begin{equation}\label{eqn k2}
K_{A_2}^{{j}}(x_2,w_2)= \dfrac{1}{\sqrt{2\pi b_2{j}}} \,e^{\frac{{j}}{2b_2}\big[a_2x_2^2-2x_2(w_2-\tau_2)-2w_2(d_2\tau_2-b_2\eta_2)+d_2(w_2^2+\tau_2^2)\big]},\,b_2\neq0\,
\end{equation}
\section{\bf Octonionic offset linear canonical transform}\label{sec 3}
 In this section we shall formally introduce the notion of proposed transform "The Octonion Offset Linear Canonical Transform($\mathbb O-$OLCT) and study its important properties like closed form representation, inversion formula, split of norm and  energy conservation. Prior to establishing the fundamental properties for the proposed transform, we shall revisit the definitions of the octonion Fourier transform($\mathbb O-$FT)\cite{1doft} and the octonion linear canonical transform($\mathbb O-$LCT)\cite{li}.
 Lets begin with definition of $\mathbb O-$FT.
\subsection{Octonion Fourier transform}\ \\
 Let   $\mu_i,i=1,2\dots,7$  denote the imaginary units in Cayley-Dickson algebra of octonions, then for an octonion-valued function $f\in L^1(\mathbb R^3,\mathbb O)$ the one dimensional  $\mathbb O-$FT \cite{1doft} is given by
 \begin{equation}\label{1doft}
 \mathcal F_{\mu_4}\{f\}(w)=\int_{\mathbb R}f(x)e^{-\mu_42\pi x w}dx,
 \end{equation}
with inversion
 \begin{eqnarray}
 \nonumber f(x)&=&\mathcal F_{\mu_4}^{-1}\{\mathcal F_{\mu_4}\{f\}\}(x)\\
 \label{1doftinverse}&=&\int_{\mathbb R} \mathcal F_{\mu_4}\{f\}(w)e^{\mu_42\pi x w}dx,
 \end{eqnarray}
 And for octonion valued function $f\in L^1(\mathbb R^3,\mathbb O)\cap L^2(\mathbb R^3,\mathbb O)$, the three dimensional  $\mathbb O-$FT \cite{20,24} is defined as
 \begin{equation}\label{3doft}
 \mathcal F_{\mu_1,\mu_2,\mu_4}\{f\}(w)=\int_{\mathbb R^3}f(x)e^{-\mu_12\pi x_1 w_1}e^{-\mu_22\pi x_2 w_2}e^{-\mu_42\pi x_3 w_3}dx,
 \end{equation}
 with inversion
 \begin{eqnarray}
 \nonumber f(x)&=&\mathcal F^{-1}_{\mu_1,\mu_2,\mu_4}\{\mathcal F_{\mu_1,\mu_2,\mu_4}\{f\}\}(x)\\
\label{3doft inverse} &=&\int_{\mathbb R^3}\mathcal F_{\mu_1,\mu_2,\mu_4}\{f\}(w)e^{-\mu_12\pi x_1 w_1}e^{-\mu_22\pi x_2 w_2}e^{-\mu_42\pi x_3 w_3}dw,
 \end{eqnarray}
 where $w=(w_1,w_2,w_3),\quad x=(x_1,x_2,x_3)\in \mathbb R^3.$\\
 The multiplication in the above integrals is done from left to right as the octonion is non-
associative. And the order of imaginary units in (\ref{3doft}) is not accidental, see\cite{4 the oft}. The  $\mathbb O-$FT of 3D octonion-valud signals follows the multiplication rules of Table-I(octonion algebra) because the octonion-valued 3D signals has octonion structure.
  \subsection{Octonion linear canonical transform}\ \\
 In 2021 Gao,W.B and Li,B.Z \cite{li} introduced linear canonical transform in octonion setting they called it the octonion linear canonical transform ($\mathbb O-$LCT) and defined it as\\
 Let $f\in L^1(\mathbb R^3,\mathbb O),$ then  the one dimensional  $\mathbb O-$LCT with respect to the uni-modular  matrix $A=(a,b,c,d)$ is given by
\begin{equation}\label{onedOLCT}
\mathcal L^{A}_{\mu_4}\{f(x)\} w)= \int_{\mathbb R} f(x)K_{A}^{\mu_4}(x,w)dx,
\end{equation}
where
 \begin{equation*}\label{eqnolctk}
K^{\mu_4}_{A}(x,w)=\dfrac{1}{\sqrt{2\pi |b|}} e^{\frac{{\mu_4}}{2b}\big[ax^2-2xw-+dw^2-\frac{\pi}{2}\big]},\quad b\neq0\,
\end{equation*}
 with inversion
 \begin{eqnarray}
  f(x)=\int_{\mathbb R}\mathcal L^{A}_{\mu_4}K_{A}^{-\mu_4}(x,w)dx,
 \end{eqnarray}
 where $K_{A}^{-\mu_4}(x,w)=K_{A^{-1}}^{\mu_4}(w,x)$ and $A^{-1}=(d,-b,-c,a).$\\

  And for octonion valued function $f\in L^1(\mathbb R^3,\mathbb O)\cap L^2(\mathbb R^3,\mathbb O)$, the three dimensional  $\mathbb O-$LCT with respect to the   matrix parameter $A_k=(a_k,b_k,c_k,d_k), $ satisfying $det(A_k)=,\quad k=1,2,3$ is defined as
 \begin{equation}\label{threedOLCT}
\mathcal L^{A_1,A_2,A_3}_{\mu_1,\mu_2,\mu_4}\{f\}(w)=\int_{\mathbb R^3}f(x) K^{\mu_1}_{A_1}(x_1,w_1)K^{\mu_2}_{A_1}(x_2,w_2)K^{\mu_4}_{A_3}(x_3,w_3)dx
\end{equation}
where $x=(x_{1},x_{2},x_3),\, w=(w_{1},w_{2},w_3),$ and  $K_{A_1}^{\mu_1}(x_1,w_1)$ , $K_{A_2}^{\mu_2}(x_2,w_2)$ and $K_{A_3}^{\mu_4}(x_3,w_3)$ are  kernel signals  given by
\begin{equation*}
K_{A_1}^{\mu_1}(x_1,w_1)=\dfrac{1}{\sqrt{2\pi |b_1|}} e^{\frac{{\mu_1}}{2b_1}\big[a_1x^2_1-2x_1w_1+d_1w^2_1-\frac{\pi}{2}\big]},\quad b_1\neq0\,
\end{equation*}
 \begin{equation*}
K_{A_2}^{\mu_2}(x_2,w_2)==\dfrac{1}{\sqrt{2\pi |b_2|}} e^{\frac{{\mu_2}}{2b_2}\big[a_2x^2_2-2x_2w_2+d_2w^2_2-\frac{\pi}{2}\big]},\quad b_2\neq0\,
\end{equation*}
and
 \begin{equation*}
K_{A_3}^{\mu_4}(x_3,w_3)=\dfrac{1}{\sqrt{2\pi |b_3|}} e^{\frac{{\mu_4}}{2b_3}\big[a_3x^2_3-2x_3w_3+d_3w^2_3-\frac{\pi}{2}\big]},\quad b_3\neq0 .
\end{equation*}
Now we are in a position to define octonion offset linear canonical transform ($\mathbb O-$OLCT).\\\\
According to the one dimensional octonion Fourier transform ($\mathbb O-$FT)(\ref{1doft}) and the  one  dimensional octonion linear canonical transform ($\mathbb O-$OLCT) (\ref{onedOLCT}) , we can obtain the definition of  one dimensional octonion offset linear canonical transform($\mathbb O-$OLCT).
\begin{definition}[One dimensional $\mathbb O-$OLCT]\label{d1oolct}
Let $f\in L^1(\mathbb R,\mathbb O),$ then one dimensional $\mathbb O-$OLCT with respect a uni-modular matrix parameter $A=(a,b,c,d,e,\tau,\eta)$ is defined as follows:

\begin{equation}\label{dOLCT}
\mathcal O^{A}_{\mu_4}\{f(x)\} w)= \int_{\mathbb R} f(x)K_{A}^{\mu_4}(x,w)dx.
\end{equation}
where

\begin{equation}\label{eqnk}
K^{\mu_4}_{A}(x,w)=\dfrac{1}{\sqrt{2\pi |b|}} e^{\frac{{\mu_4}}{2b}\big[ax^2-2x(w-\tau)-2w(d\tau-b\eta)+d(w^2+\tau^2)-\frac{\pi}{2}\big]},\quad b\neq0\,
\end{equation}
\end{definition}

The following lemma gives the relationship of one dimensional $\mathbb O-$OLCT and one dimensional $\mathbb O-$FT  of an octonion-valued signals.
\begin{lemma}\label{lemrelation}The one dimensional  $\mathbb O-$OLCT  of a signal $f\in L^1(\mathbb R,\mathbb O)$ can be reduced to one dimensional $\mathbb O-$FT (\ref{1doft}) as
\begin{equation}\label{relation}
\mathcal O^{A}_{\mu_4}\{f(x)\} w)=\frac{1}{\sqrt{2\pi|b|}}\mathcal F_{\mu_4}\left\{f(x)e^{[\frac{a}{2b}x^2+\frac{1}{b}x\tau]}\right\}\left(\frac{w}{2\pi|b|}\right)e^{\mu_4[\frac{d}{2b}(w^2+\tau^2)+\frac{w}{b}(b\eta-d\tau)-\frac{\pi}{4}]}
\end{equation}
where  $b\ne 0$ and
\begin{equation}
\mathcal F_{\mu_4}\{f(x)\}(w)=\int_{\mathbb R}f(x)e^{-\mu_4 2\pi x w}dx
\end{equation}
represents the one dimensional $\mathbb O-$FT  of an octonion-valued signal $f(x)$.
\end{lemma}

By applying  Lemma \ref{lemrelation} and (\ref{1doftinverse}), we get the formula for the inversion of one dimensional $\mathbb O-$OLCT, which is given below as a theorem.

\begin{theorem}\label{th1dolctinverse} Let $f$ be an octonion-valued signal $\in L^1(\mathbb R,\mathbb O).$ Then the inversion formula of the one dimensional $\mathbb O-$OLCT is
\begin{equation}\label{eqninverse1doolct}
f(x)=\int_{\mathbb R}\mathcal O^A_{\mu_4}[f](w)\overline{K^{\mu_4}_{A}(x,w)}dw,
\end{equation}
where $\overline{K^{\mu_4}_{A}(x,w)}=K^{-\mu_4}_{A}(x,w)$ and $b\ne 0.$
\end{theorem}

On  replacing  the complex unit $i$ in the ordinary offset linear canonical transform by the imaginary units in the
octonions, the three dimensional octonion offset linear canonical transform($\mathbb O-$OLCT) could be defined as
\begin{definition}\label{defolct}{\bf(Three dimensional $\mathbb O-$OLCT)}\ \\
Let $A_k =\left[\begin{array}{cccc}a_k & b_k &| & \tau_k\\c_k & d_k &| & \eta_k \\\end{array}\right] $, be a matrix parameter such that $a_k$, $b_k$, $c_k$, $d_k$, $p_k$, $q_k \in \mathbb R$ and $ a_kd_k-b_kc_k=1,$ for $k=1,2,3.$ The three dimensional $\mathbb O-$OLCT of an octonion-valued signal  $f$ over $\mathbb R^3,$ is given by
\begin{equation}\label{eqnolct}
\mathcal O^{A_1,A_2,A_3}_{\mu_1,\mu_2,\mu_4}\{f\}(w)=\int_{\mathbb R^3}f(x) K^{\mu_1}_{A_1}(x_1,w_1)K^{\mu_2}_{A_1}(x_2,w_2)K^{\mu_4}_{A_3}(x_3,w_3)dx
\end{equation}
where $x=(x_{1},x_{2},x_3),\, w=(w_{1},w_{2},w_3),$ and  $K_{A_1}^{\mu_1}(x_1,w_1)$ , $K_{A_2}^{\mu_2}(x_2,w_2)$ and $K_{A_3}^{\mu_4}(x_3,w_3)$are  kernel signals  given by
\begin{equation}
K_{A_1}^{\mu_1}(x_1,w_1)==\dfrac{1}{\sqrt{2\pi |b_1|}} e^{\frac{{\mu_1}}{2b_1}\big[a_1x^2_1-2x_1(w_1-\tau_1)-2w_1(d_1\tau_1-b_1\eta_1)+d_1(w^2_1+\tau^2_1)-\frac{\pi}{2}\big]},\quad b\neq0\,
\end{equation}
\begin{equation}
K_{A_2}^{\mu_2}(x_2,w_2)=\dfrac{1}{\sqrt{2\pi |b_2|}} e^{\frac{{\mu_2}}{2b_2}\big[a_2x^2_2-2x_2(w_2-\tau_2)-2w_2(d_2\tau_2-b_2\eta_2)+d_2(w^2_2+\tau^2_2)-\frac{\pi}{2}\big]},\quad b\neq0\,
\end{equation}
and
\begin{equation}
K_{A_3}^{\mu_4}(x_3,w_3)==\dfrac{1}{\sqrt{2\pi |b_3|}} e^{\frac{{\mu_4}}{2b_3}\big[a_3x^2_3-2x_3(w_3-\tau_3)-2w_3(d_3\tau_3-b_3\eta_3)+d_3(w^2_3+\tau^2_3)-\frac{\pi}{2}\big]},\quad b\neq0\,.
\end{equation}
\end{definition}

Note that, for $b_k=0,k=1,2,3$, the $\mathbb O-$OLCT boils down to chirp multiplication operator and it is of no particular intrest for our objective in this work. Hence for the sake of braveity, we always set $b_k\ne 0$ in the paper unless stated otherwise.\\\\
Also note that,the kernels with imaginary units $\mu_1,\mu_2,\mu_4$ are octonion-valued and does not  reduce to the
quaternion cases, thus the present integral transform is more interesting
 and complicated.\\
Now we will  obtain the closed-form representation of $\mathbb O-$OLCT defined in (\ref{eqnolct}), let us begin by setting
 \begin{equation}\label{xi}\xi_k={\frac{{1}}{2b_k}\big[a_kx^2_k-2x_k(w_k-\tau_k)-2w_k(d_k\tau_k-b_k\eta_k)+d_k(w^2_k+\tau^2_k)-\frac{\pi}{2}\big]},\quad k=1,2,3.\end{equation}
Thus,
\begin{eqnarray}\label{eulerproduct}
 \nonumber K_{A_1}^{\mu_1}(x_1,w_1)K_{A_2}^{\mu_2}(x_2,w_2)K_{A_3}^{\mu_4}(x_3,w_3)&=&\frac{1}{2\pi\sqrt{2\pi|b_1b_2b_3|}}e^{\mu_1\xi_1}e^{\mu_2\xi_2}e^{\mu_4\xi_3}\\
\nonumber&=&\frac{1}{2\pi\sqrt{2\pi|b_1b_2b_3|}}(c_1+\mu_1s_1)(c_2+\mu_2s_2)(c_3+\mu_4s_3)\\
\nonumber&=&\frac{1}{2\pi\sqrt{2\pi|b_1b_2b_3|}}(c_1c_2c_3+s_1c_2c_3\mu_1+c_1s_2c_3\mu_2\\
\nonumber&\quad+& s_1s_2c_3\mu_3+c_1c_2s_3\mu_4+s_1c_2s_3\mu_5+c_1s_2s_3\mu_6+s_1s_2s_3\mu_7),\\
\end{eqnarray}
where $c_k=\cos\xi_k$ and $s_k=\sin\xi_k,\quad k=1,2,3.$\\
Using (\ref{eulerproduct}) in (\ref{eqnolct}) we get closed-form representation of $\mathbb O-$OLCT given by following lemma.\\

\begin{lemma}[Closed-form representation]\label{lemclosed}
The $\mathbb O-$OLCT of a three dimensional signal $f:\mathbb R^3\longrightarrow\mathbb O$ has the closed-form representation:
\begin{equation}\label{eqnclosed}
\mathcal O^{A_1,A_2,A_3}_{\mu_1,\mu_2,\mu_3}\{f\}(w)={\mathbf\Phi}_0(w)+{\mathbf\Phi}_1(w)+{\mathbf\Phi}_2(w)+{\mathbf\Phi}_3(w)
+{\mathbf\Phi}_4(w)+{\mathbf\Phi}_5(w)+{\mathbf\Phi}_6(w)+{\mathbf\Phi}_7(w)
\end{equation}
\end{lemma}
where we put the integrals
\begin{eqnarray*}
{\mathbf\Phi}_0(w)=\int_{\mathbb R^3}f_{eee}(x)\frac{1}{2\pi\sqrt{2\pi|b_1b_2b_3|}}c_1c_2c_3dx,
\end{eqnarray*}
and,
\begin{eqnarray*}
{\mathbf\Phi}_1(w)=\int_{\mathbb R^3}f_{oee}(x)\frac{\mu_1}{2\pi\sqrt{2\pi|b_1b_2b_3|}}s_1c_2c_3dx,
\end{eqnarray*}
\begin{eqnarray*}
{\mathbf\Phi}_2(w)=\int_{\mathbb R^3}f_{eoe}(x)\frac{\mu_2}{2\pi\sqrt{2\pi|b_1b_2b_3|}}c_1s_2c_3dx,
\end{eqnarray*}
\begin{eqnarray*}
{\mathbf\Phi}_3(w)=\int_{\mathbb R^3}f_{ooe}(x)\frac{\mu_3}{2\pi\sqrt{2\pi|b_1b_2b_3|}}s_1s_2c_3dx,
\end{eqnarray*}
\begin{eqnarray*}
{\mathbf\Phi}_4(w)=\int_{\mathbb R^3}\frac{\mu_4}{2\pi\sqrt{2\pi|b_1b_2b_3|}}f_{eeo}(x)c_1c_2s_3dx,
\end{eqnarray*}
\begin{eqnarray*}
{\mathbf\Phi}_5(w)=\int_{\mathbb R^3}f_{oeo}(x)\frac{\mu_5}{2\pi\sqrt{2\pi|b_1b_2b_3|}}s_1c_2s_3dx,
\end{eqnarray*}
\begin{eqnarray*}
{\mathbf\Phi}_6(w)=\int_{\mathbb R^3}f_{eoo}(x)\frac{\mu_6}{2\pi\sqrt{2\pi|b_1b_2b_3|}}c_1s_2s_3dx,
\end{eqnarray*}
\begin{eqnarray*}
{\mathbf\Phi}_7(w)=\int_{\mathbb R^3}f_{ooo}(x)\frac{\mu_7}{2\pi\sqrt{2\pi|b_1b_2b_3|}}s_1s_2s_3dx,
\end{eqnarray*}

where the functions $f_{xyz},x,y,z\in\{e,o\},$ are eight components of $f$ of different parity with respect appropriate variable, for example $f_{oeo}(x)$ is odd with respect to $x_1$, even with respect $x_2$ and odd with respect to $x_3.$

Under suitable conditions, the original octonion-valued signal $f$ can be reconstructed from $\mathbb O-$OLCT  by its inverse transform.
\begin{definition}\label{inverse} The inverse $\mathbb O-$OLCT  of signal $g:\mathbb R^3\longrightarrow \mathbb O$ is defined by
\begin{equation}
\left\{O^{A_1,A_2,A_3}_{\mu_1,\mu_2,\mu_4}\right\}^{-1}\{g\}(x)=\int_{\mathbb R^3}f(w)\overline{K_{A_3}^{\mu_4}(x_1,w_1)} \overline{K_{A_2}^{\mu_2}(x_2,w_2)}\overline{K_{A_1}^{\mu_1}(x_3,w_3)}dw
 \end{equation}
 It has the closed-form representation:
  \begin{equation}
  \left\{O^{A_1,A_2,A_3}_{\mu_1,\mu_2,\mu_4}\right\}^{-1}\{g\}(x)={\mathbf\Psi}_0(w)+{\mathbf\Psi}_1(w)+{\mathbf\Psi}_2(w)+{\mathbf\Psi}_3(w)
+{\mathbf\Psi}_4(w)+{\mathbf\Psi}_5(w)+{\mathbf\Psi}_6(w)+{\mathbf\Psi}_7(w)
   \end{equation}
   where we put the integrals
\begin{eqnarray*}
{\mathbf\Psi}_0(w)=\int_{\mathbb R^3}f_{eee}(x)\frac{1}{2\pi\sqrt{2\pi|b_1b_2b_3|}}c_1c_2c_3dx,
\end{eqnarray*}
and,
\begin{eqnarray*}
{\mathbf\Psi}_1(w)=\int_{\mathbb R^3}f_{oee}(x)\frac{-\mu_1}{2\pi\sqrt{2\pi|b_1b_2b_3|}}s_1c_2c_3dx,
\end{eqnarray*}
\begin{eqnarray*}
{\mathbf\Psi}_2(w)=\int_{\mathbb R^3}f_{eoe}(x)\frac{-\mu_2}{2\pi\sqrt{2\pi|b_1b_2b_3|}}c_1s_2c_3dx,
\end{eqnarray*}
\begin{eqnarray*}
{\mathbf\Psi}_3(w)=\int_{\mathbb R^3}f_{ooe}(x)\frac{-\mu_3}{2\pi\sqrt{2\pi|b_1b_2b_3|}}s_1s_2c_3dx,
\end{eqnarray*}
\begin{eqnarray*}
{\mathbf\Psi}_4(w)=\int_{\mathbb R^3}f_{eeo}(x)\frac{-\mu_4}{2\pi\sqrt{2\pi|b_1b_2b_3|}}c_1c_2s_3dx,
\end{eqnarray*}
\begin{eqnarray*}
{\mathbf\Psi}_5(w)=\int_{\mathbb R^3}f_{oeo}(x)\frac{-\mu_5}{2\pi\sqrt{2\pi|b_1b_2b_3|}}s_1c_2s_3dx,
\end{eqnarray*}
\begin{eqnarray*}
{\mathbf\Psi}_6(w)=\int_{\mathbb R^3}f_{eoo}(x)\frac{-\mu_6}{2\pi\sqrt{2\pi|b_1b_2b_3|}}c_1s_2s_3dx,
\end{eqnarray*}
\begin{eqnarray*}
{\mathbf\Psi}_7(w)=\int_{\mathbb R^3}f_{ooo}(x)\frac{-\mu_7}{2\pi\sqrt{2\pi|b_1b_2b_3|}}s_1s_2s_3dx,
\end{eqnarray*}
\end{definition}

The simple way to define inverse of $\mathbb O-$OLCT is to introduce  Inversion theorem.

\begin{theorem}[ Inversion for three dimensional $\mathbb O-$OLCT ]\label{thminverse3dolct}\ \\
Every octonion-valued signal $f:\mathbb R^3\longrightarrow \mathbb O$ can be reconstructed by the formula
\begin{eqnarray*}\label{eqninv3dolct}
\nonumber f(x)&=&\left\{\mathcal O^{A_1,A_2,A_3}_{\mu_1,\mu_2,\mu_4}\right\}^{-1}[\mathcal O^{A_1,A_2,A_3}_{\mu_1,\mu_2,\mu_4}\{f\}](x)\\
\nonumber&=&\int_{\mathbb R^3}\mathcal O^{A_1,A_2,A_3}_{\mu_1,\mu_2,\mu_4}\{f\}(w)K_{A_3}^{-\mu_4}(x_1,w_1) K_{A_2}^{-\mu_2}(x_2,w_2)K_{A_1}^{-\mu_1}(x_3,w_3)dw,\\
\end{eqnarray*}
\begin{proof}
By using the definition of QOLCT\cite{OWN1}, the one dimensional $\mathbb O-$OLCT and three dimensional $\mathbb O-$OLCT the proof of theorem \ref{thminverse3dolct} follows.
\end{proof}
\end{theorem}
For the clarity of the formulas we denote $x_{l,m,n}=(lx_1,mx_2,nx_3),$ $l,m,n\in\{+,-\},$ i.e. $x_{+-+}=(x_1,-x_2,x_3)$ and denoting the even and odd part of a function $f(x)$ by $f_e(x)$ and $f_o(x)$ where $f_e=(f(x_{+++})+f(x_{++-}))/2$ which is only
even in the third variable $x_3.$  Similarly, $f_o=(f(x_{+++})-f(x_{++-}))/2.$\\
\begin{lemma}[Norm split]\label{lemoctnorm}Let $f:\mathbb R^3\longrightarrow \mathbb O$  beoctonion-valued signal and $\mathcal O^{A_1,A_2,A_3}_{\mu_1,\mu_2,\mu_4}\{f\}(w)$ be the  $\mathbb O-$OLCT of $f,$ then
\begin{eqnarray*}\label{eqnoctnorm}
\|\mathcal O^{A_1,A_2,A_3}_{\mu_1,\mu_2,\mu_4}\{f\}(w)\|^2&=&\frac{1}{2\pi|b_3|}\left(\|\mathcal O^{A_1,A_2}_{\mu_1,\mu_2}\{g_e\}(w)\|^2+\|\mathcal O^{A_1,A_2}_{\mu_1,\mu_2}\{h_o\}(w)\|^2\right.\\
&&+\left.\|\mathcal O^{A_1,A_2}_{\mu_1,\mu_2}\{h_e\}(w)\|^2+\|\mathcal O^{A_1,A_2}_{\mu_1,\mu_2}\{g_o\}(w)\|^2\right)\\
\end{eqnarray*}

\begin{proof}
From (\ref{ofun}) every octonion-valued signal $f$ has explicit form  $f=g+h\mu_4$ where $g,h\in\mathbb H.$  And  by the even and odd part we further express the $\mathbb O-$OLCT in  (\ref{eqnolct})  as follows,
\begin{eqnarray}\label{qolct}
\nonumber\mathcal O^{A_1,A_2,A_3}_{\mu_1,\mu_2,\mu_4}\{f\}(w)&=&\int_{\mathbb R^3}g(x)K_{A_1}^{\mu_1}(x_1,w_1)K_{A_2}^{\mu_2}(x_2,w_2)K_{A_3}^{\mu_4}(x_3,w_3)dx\\
\nonumber&&+\int_{\mathbb R^3}h(x)K_{A_1}^{-\mu_1}(x_1,w_1)K_{A_2}^{-\mu_2}(x_2,w_2)\mu_4K_{A_3}^{\mu_4}(x_3,w_3)dx\\
\nonumber&=&\frac{1}{\sqrt{2\pi|b_3|}}\int_{\mathbb R^3}g_e(x)K_{A_1}^{\mu_1}(x_1,w_1)K_{A_2}^{\mu_2}(x_2,w_2)c_3dx\\
\nonumber&&+\frac{1}{\sqrt{2\pi|b_3|}}\int_{\mathbb R^3}h_o(x)K_{A_1}^{-\mu_1}(x_1,w_1)K_{A_2}^{-\mu_2}(x_2,w_2)s_3dx\\
\nonumber&&+\left(\frac{1}{\sqrt{2\pi|b_3|}}\int_{\mathbb R^3}h_e(x)K_{A_1}^{-\mu_1}(x_1,w_1)K_{A_2}^{-\mu_2}(x_2,w_2)c_3dx\right.\\
\nonumber&&-\left.\frac{1}{\sqrt{2\pi|b_3|}}\int_{\mathbb R^3}g_o(x)K_{A_1}^{\mu_1}(x_1,w_1)K_{A_2}^{\mu_2}(x_2,w_2)s_3dx\right)\mu_4.\\
\end{eqnarray}
With $c_k=\cos\xi_k $ and $s_k=\sin\xi_k$  where $\xi_k$ is given in (\ref{xi}).
From (\ref{qolct}) it is clear that $\mathbb O-$OLCT can be divided into four QOLCTs. Thus the norm of $\mathbb O-$OLCT splits into four norms of quaternion functions as:
\begin{eqnarray}\label{eqnoctnorm1}
\nonumber\|\mathcal O^{A_1,A_2,A_3}_{\mu_1,\mu_2,\mu_4}\{f\}(w)\|^2_2&=&\frac{1}{{2\pi|b_3|}}\left\|\int_{\mathbb R^3}g_e(x)K_{A_1}^{\mu_1}(x_1,w_1)K_{A_2}^{\mu_2}(x_2,w_2)c_3dx\right\|^2_2\\
\nonumber&&+\frac{1}{{2\pi|b_3|}}\left\|\int_{\mathbb R^3}h_o(x)K_{A_1}^{-\mu_1}(x_1,w_1)K_{A_2}^{-\mu_2}(x_2,w_2)s_3dx\right\|^2_2\\
\nonumber&&+\frac{1}{{2\pi|b_3|}}\left\|\int_{\mathbb R^3}h_e(x)K_{A_1}^{-\mu_1}(x_1,w_1)K_{A_2}^{-\mu_2}(x_2,w_2)c_3dx\right\|^2_2\\
\nonumber&&+\frac{1}{{2\pi|b_3|}}\left\|\int_{\mathbb R^3}g_o(x)K_{A_1}^{\mu_1}(x_1,w_1)K_{A_2}^{\mu_2}(x_2,w_2)s_3dx \right\|^2_2\\
\end{eqnarray}
where the  equality is by the fact that $f_e$ and $f_o$ are orthogonal in $L^2$ inner product.
\begin{eqnarray*}
\nonumber\|\mathcal O^{A_1,A_2,A_3}_{\mu_1,\mu_2,\mu_4}\{f\}(w)\|^2_2&=&\frac{1}{2\pi|b_3|}\left(\|\mathcal O^{A_1,A_2}_{\mu_1,\mu_2}\{g_e\}(w)\|^2_2+\|\mathcal O^{A_1,A_2}_{\mu_1,\mu_2}\{h_o\}(w)\|^2_2\right.\\
\nonumber&&+\left.\|\mathcal O^{A_1,A_2}_{\mu_1,\mu_2}\{h_e\}(w)\|^2_2+\|\mathcal O^{A_1,A_2}_{\mu_1,\mu_2}\{g_o\}(w)\|^2_2\right)\\
\end{eqnarray*}

Which completes the proof.
\end{proof}
\end{lemma}

Again denoting the quaternion form of $f$ by $f=g+h\mu_4$, we have
\begin{eqnarray}\label{evenoddfun}
\nonumber\|f\|^2_2&=&\|g+h\mu_4\|^2_2\\
\nonumber&=&\|g\|^2+\|h\|^2\\
\nonumber&=&\|g_e\|^2+\|g_o\|^2+\|h_e\|^2+\|h_o\|^2\\
\end{eqnarray}

Now by the Plancherel theorem for the QOLCT \cite{OWN1}, we have
\begin{equation}\label{planch}
\left\|\int_{\mathbb R^3}g_e(x)K_{A_1}^{\mu_1}(x_1,w_1)K_{A_2}^{\mu_2}(x_2,w_2)c_3dx\right\|^2_2=\left\|g_e(x)\right\|^2_2
\end{equation}
Similarly results hold for the remaining three norms.\\

On applying (\ref{evenoddfun}) and (\ref{planch}) in lemma \ref{lemoctnorm}
we have proved the following Energy conservation relation for $\mathbb O-$OLCT.
\begin{theorem}[Energy conservation]\label{thenergy}\ \\
Let $f:\mathbb R^3\longrightarrow \mathbb O$ be a continuous and square integrable octonion-valued signal function. Then we have
\begin{equation}\label{eqnenergy}
\|\mathcal O^{A_1,A_2,A_3}_{\mu_1,\mu_2,\mu_4}\{f\}(w)\|^2=\frac{1}{2\pi|b_3|}\|f\|^2_2
\end{equation}
\end{theorem}
Now we move towards our main section in which we present some uncertainty principles associated with $\mathbb O-$OLCT.

\section{\bf Uncertainty principles for  $\mathbb O-$OLCT}\label{sec 4}
We know that in signal processing there are  different types of uncertainty principles in the QFT, QLCT and QOLCT domains. Recently in \cite{li} authors investigate Heisenberg's uncertainty principle and Donoho-Stark's uncertainty principle for $\mathbb O-$LCT.
Considering that the $\mathbb O-$OLCT is a generalized
version of the $\mathbb O-$FT or $\mathbb O-$LCT, it is natural and interesting to study  uncertainty principles of a octonion-valued function and its
$\mathbb O-$OLCT. So in this section we  shall investigate some uncertainty principles for $\mathbb O-$OLCT.
\subsection{\bf Pitt's inequality and the logarithmic uncertainty principle }\ \\

Here we prove the sharp Pitt's inequality for the $\mathbb O-$OLCT and derive the associated logarithmic uncertainty inequality. The proof of the Pitt’s inequality heavily depends on the QOLCT. So we first present following Pitt's inequality for QOLCT.
\begin{lemma}[ Pitt's inequality for QOLCT]\label{lem pitt qolct}\cite{gen} \ \\
For $f\in \mathcal S(\mathbb R^2,\mathbb H),$ and     $0\le \alpha <2$,
\begin{equation}\label{eqnpitt qolct}
\int_{{\mathbb R}^2}{ {\left|\frac{ w}{b}\right|}^{-\alpha }}{\left|{\mathcal O}^{i,j}_{A_1,A_2}\left\{f(x)\right\}\left( w\right)\right|}^2d w\  \le \frac{ C_{\alpha }}{4\pi^2} \int_{{\mathbb R}^2}{ {\left| x\right|}^{\alpha }}{\left|f( x)\right|}^2\ d x.
\end{equation}
With $C_{\alpha }:=\frac{{4\pi }^2}{2^{\alpha }}{{ [}\Gamma (\frac{2-\alpha }{4})/\Gamma (\frac{2+\alpha }4)]}^2$, and $\Gamma \left(.\right)$\ is the Gamma function and ${ \mathcal S}({\mathbb R}^2,{\mathbb H})$ denotes the
Schwartz space.
\end{lemma}
\begin{theorem}[Pitt's inequality for the $\mathbb O-$OLCT]\label{thmpitt} For $f\in \mathcal S(\mathbb R^3,\mathbb O),$ and     $0\le \alpha <3$ and  under the assumptions of lemma \ref{lem pitt qolct}, we have
\begin{equation}\label{eqnpittolct}
\int_{{\mathbb R}^3}{ {\left|\frac{ w}{b}\right|}^{-\alpha }}{\left|{\mathcal O}^{A_1,A_2,A_3}_{\mu_1,\mu_2,\mu_3}\left\{f(x)\right\}\left( w\right)\right|}^2d w\  \le \frac{C_{\alpha }}{8\pi^3|b_3|}\int_{{\mathbb R}^3}{ {\left| x\right|}^{\alpha }}|f(x)|^2dx.
\end{equation}
\end{theorem}
\begin{proof}
 We have split $\mathbb O-$OLCT into four QOLCT in lemma \ref{lemoctnorm}, therefore
\begin{eqnarray}\label{ppp}
\nonumber\int_{{\mathbb R}^3}{ {\left|\frac{ w}{b}\right|}^{-\alpha }}{\left|{\mathcal O}^{A_1,A_2,A_3}_{\mu_1,\mu_2,\mu_3}\left\{f(x)\right\}\left( w\right)\right|}^2d w&=&\frac{1}{2\pi|b_3|}\left(\int_{{\mathbb R}^3}{ {\left|\frac{ w}{b}\right|}^{-\alpha }}|\mathcal O^{A_1,A_2}_{\mu_1,\mu_2}\{g_e\}(w)|^2d w\right.\\
\nonumber&&+\left.\int_{{\mathbb R}^3}{ {\left|\frac{ w}{b}\right|}^{-\alpha }}|\mathcal O^{A_1,A_2}_{\mu_1,\mu_2}\{g_o\}(w)|^2d w\right.\\
\nonumber&&+\left.\int_{{\mathbb R}^3}{ {\left|\frac{ w}{b}\right|}^{-\alpha }}|\mathcal O^{A_1,A_2}_{\mu_1,\mu_2}\{h_e\}(w)|^2d w\right.\\
\label{ppp}&&+\left.\int_{{\mathbb R}^3}{ {\left|\frac{ w}{b}\right|}^{-\alpha }}|\mathcal O^{A_1,A_2}_{\mu_1,\mu_2}\{g_e\}(w)|^2d w\right)
\end{eqnarray}
By the Pitt's inequality for QOLCT (\ref{eqnpitt qolct}), we have

\begin{equation}
\int_{{\mathbb R}^3}{ {\left|\frac{ w}{b}\right|}^{-\alpha }}{\left|{\mathcal O}^{i,j}_{A_1,A_2}\left\{g_e(x)\right\}\left( w\right)\right|}^2d w\  \le \frac{ C_{\alpha }}{4\pi^2} \int_{{\mathbb R}^3}{ {\left| x\right|}^{\alpha }}{\left|f( x)\right|}^2\ d x.
\end{equation}
As similar inequalities hold for the remaining three terms. Collecting all and inserting in (\ref{ppp}), we get
\begin{eqnarray*}
\nonumber\int_{{\mathbb R}^3}{ {\left|\frac{ w}{b}\right|}^{-\alpha }}{\left|{\mathcal O}^{A_1,A_2,A_3}_{\mu_1,\mu_2,\mu_3}\left\{f(x)\right\}\left( w\right)\right|}^2d w &\le&\frac{C_{\alpha }}{8\pi^3|b_3|}\left(\int_{{\mathbb R}^3}{ {\left| x\right|}^{\alpha }}{\left|g_e\right|}^2\ d x
+\int_{{\mathbb R}^3}{ {\left| x\right|}^{\alpha }}{\left|g_o( x)\right|}^2\ d x\right.\\
&&+\left.\int_{{\mathbb R}^3}{ {\left| x\right|}^{\alpha }}{\left|h_e( x)\right|}^2\ d x
+\int_{{\mathbb R}^3}{ {\left| x\right|}^{\alpha }}{\left|h_o( x)\right|}^2\ d x\right)\\
&=&\frac{C_{\alpha }}{8\pi^3|b_3|}\int_{{\mathbb R}^3}{ {\left| x\right|}^{\alpha }}\left(|g_e( x)|^2+|g_o( x)|^2+|h_e( x)|^2+|h_o( x)|^2\right)dx.\\
&=&\frac{C_{\alpha }}{8\pi^3|b_3|}\int_{{\mathbb R}^3}{ {\left| x\right|}^{\alpha }}|f(x)|^2dx
\end{eqnarray*}
where last equality occurs because of (\ref{evenoddfun}).\\
Which completes the proof.
\end{proof}
Here $C_\alpha$ can't be smaller any more. It is equal to the  ordinary complex and the quaternion  cases. Thus, the inequality is sharp.
 If $\alpha=0$, it changes to equality, at $\alpha=0$  differentiating the sharp Pitt’s inequalities led to the following logarithmic uncertainty inequality for the $\mathbb O-$OLCT.

\begin{theorem}[Logarithmic uncertainty principle for the $\mathbb O-$OLCT]\label{thlogoolct} Let $f\in \mathcal S(\mathbb R^3,\mathbb O),$ then the following inequality is satisfied:
\begin{equation}\label{eqnlogoolct}
2\pi|b_3|\int_{{\mathbb R}^2}\ln\left|\frac{w}{b}\right|\left|{\mathcal O}^{A_1,A_2,A_3}_{\mu_1,\mu_2,\mu_3}\{f\}( w)\right|^2d w + \int_{{\mathbb R}^2}\ln|x| |f(x)|^2dx\ge D \int_{{\mathbb R}^2} |f(x)|^2dx
\end{equation}
with $D=\ln(2)+\Gamma'(\frac{1}{2})/\Gamma(\frac{1}{2}).$
\end{theorem}
\begin{proof}
 Following the  procedure of theorem 4.11 in \cite{gen} we will get desired result.

Alternatively, we can prove Logarithmic uncertainty principle for the $\mathbb O-$OLCT from the Logarithmic uncertainty principle for the QOLCT \cite{gen}.\\
\begin{lemma}[Logarithmic uncertainty principle for the QOLCT]\label{lem log qolct}\cite{gen}\ \\
Let$f\in \mathcal S(\mathbb R^2,\mathbb H),$ then
\begin{equation}\label{eqnlogoolct}
\int_{{\mathbb R}^2}\ln\left|\frac{w}{b}\right|\left|{\mathcal O}^{A_1,A_2}_{\mu_1,\mu_2}\{f\}( w)\right|^2d w + \int_{{\mathbb R}^2}\ln|x| |f(x)|^2dx\ge D \int_{{\mathbb R}^2} |f(x)|^2dx
\end{equation}
with $D=\ln(2)+\Gamma'(\frac{1}{2})/\Gamma(\frac{1}{2}).$
\end{lemma}
{\bf Proof of theorem \ref{thlogoolct}}\\
By lemma \ref{lemoctnorm}, $\mathbb O-$OLCT can be written in split quaternion form as
\begin{eqnarray*}
\|\mathcal O^{A_1,A_2,A_3}_{\mu_1,\mu_2,\mu_4}\{f\}(w)\|^2&=&\frac{1}{2\pi|b_3|}\left(\|\mathcal O^{A_1,A_2}_{\mu_1,\mu_2}\{g_e\}(w)\|^2+\|\mathcal O^{A_1,A_2}_{\mu_1,\mu_2}\{h_o\}(w)\|^2\right.\\
&&+\left.\|\mathcal O^{A_1,A_2}_{\mu_1,\mu_2}\{h_e\}(w)\|^2+\|\mathcal O^{A_1,A_2}_{\mu_1,\mu_2}\{g_o\}(w)\|^2\right).\\
\end{eqnarray*}
Therefore
\begin{eqnarray*}
\int_{\mathbb R^2}\ln\left|\frac{w}{b}\right|\|\mathcal O^{A_1,A_2,A_3}_{\mu_1,\mu_2,\mu_4}\{f\}(w)\|^2dw&=&\frac{1}{2\pi|b_3|}\int_{\mathbb R^2}\ln\left|\frac{w}{b}\right|\left(\|\mathcal O^{A_1,A_2}_{\mu_1,\mu_2}\{g_e\}(w)\|^2+\|\mathcal O^{A_1,A_2}_{\mu_1,\mu_2}\{h_o\}(w)\|^2\right.\\
&&+\left.\|\mathcal O^{A_1,A_2}_{\mu_1,\mu_2}\{h_e\}(w)\|^2+\|\mathcal O^{A_1,A_2}_{\mu_1,\mu_2}\{g_o\}(w)\|^2\right)dw.\\
\end{eqnarray*}
Implies
\begin{eqnarray}\label{logeqn split}
\nonumber&&2\pi|b_3|\int_{\mathbb R^2}\ln\left|\frac{w}{b}\right|\|\mathcal O^{A_1,A_2,A_3}_{\mu_1,\mu_2,\mu_4}\{f\}(w)\|^2dw\\\
\nonumber&&\quad=\int_{\mathbb R^2}\ln\left|\frac{w}{b}\right|\|\mathcal O^{A_1,A_2}_{\mu_1,\mu_2}\{g_e\}(w)\|^2dw+\int_{\mathbb R^2}\ln\left|\frac{w}{b}\right|\|\mathcal O^{A_1,A_2}_{\mu_1,\mu_2}\{h_o\}(w)\|^2dw.\\\
\nonumber&&\quad\quad+\int_{\mathbb R^2}\ln\left|\frac{w}{b}\right|\|\mathcal O^{A_1,A_2}_{\mu_1,\mu_2}\{h_e\}(w)\|^2dw+\int_{\mathbb R^2}\ln\left|\frac{w}{b}\right|\|\mathcal O^{A_1,A_2}_{\mu_1,\mu_2}\{g_o\}(w)\|^2dw\\\
\end{eqnarray}
Also by virtue of  (\ref{evenoddfun}), we can write
\begin{eqnarray}\label{logeqn split1}
\nonumber\int_{\mathbb R^2}\ln|x||f(x)|^2dx&=&\int_{\mathbb R^2}\ln|x||g_e(x)|^2dx+\int_{\mathbb R^2}\ln|x||g_o(x)|^2dx\\
\nonumber&&+\int_{\mathbb R^2}\ln|x||h_e(x)|^2dx+\int_{\mathbb R^2}\ln|x||h_o(x)|^2dx.\\
\end{eqnarray}
By  logarithmic uncertainty principle for the QOLCT given in lemma \ref{lem log qolct}, we get
\begin{equation}\label{l1}
\int_{\mathbb R^2}\ln\left|\frac{w}{b}\right|\|\mathcal O^{A_1,A_2}_{\mu_1,\mu_2}\{g_e\}(w)\|^2+\int_{\mathbb R^2}\ln|x||g_e(x)|^2dx\ge D\int_{\mathbb R^2}|g_e(x)|^2dx.
\end{equation}
Similarly  ,
\begin{equation}\label{l2}
\int_{\mathbb R^2}\ln\left|\frac{w}{b}\right|\|\mathcal O^{A_1,A_2}_{\mu_1,\mu_2}\{g_e\}(w)\|^2+\int_{\mathbb R^2}\ln|x||g_e(x)|^2dx\ge D\int_{\mathbb R^2}|g_o(x)|^2dx,
\end{equation}
\begin{equation}\label{l3}
\int_{\mathbb R^2}\ln\left|\frac{w}{b}\right|\|\mathcal O^{A_1,A_2}_{\mu_1,\mu_2}\{g_e\}(w)\|^2+\int_{\mathbb R^2}\ln|x||g_e(x)|^2dx\ge D\int_{\mathbb R^2}|h_e(x)|^2dx.
\end{equation}
\begin{equation}\label{l4}
\int_{\mathbb R^2}\ln\left|\frac{w}{b}\right|\|\mathcal O^{A_1,A_2}_{\mu_1,\mu_2}\{g_e\}(w)\|^2+\int_{\mathbb R^2}\ln|x||g_e(x)|^2dx\ge D\int_{\mathbb R^2}|h_o(x)|^2dx.
\end{equation}
Collecting all equations (\ref{l1}, \ref{l2}, \ref{l3}, \ref{l4}) and making use of (\ref{logeqn split}, \ref{logeqn split1}), we obtain the desired result
\begin{equation}
2\pi|b_3|\int_{{\mathbb R}^2}\ln\left|\frac{w}{b}\right|\left|{\mathcal O}^{A_1,A_2,A_3}_{\mu_1,\mu_2,\mu_3}\{f\}( w)\right|^2d w + \int_{{\mathbb R}^2}\ln|x| |f(x)|^2dx\ge D \int_{{\mathbb R}^2} |f(x)|^2dx.
\end{equation}
Which completes the proof
\end{proof}
\subsection{Hausdorff-Young inequality  for  $\mathbb O-$OLCT}\ \\
In this subsection we will establish Hausdorff-Young inequality which is very important in signal processing. This inequality will be helpful for researchers in establishing  Shannon's entropy uncertainty relation.

\begin{lemma}[Hausdorff-Young inequality for QOLCT]\label{lem hdrfolct}\cite{ucp 2sd qolct} For $1\le p\le 2$ and $\frac{1}{p}+\frac{1}{q}=1,$ we have
\begin{equation}\label{eqnhdrof olct}
\|\mathcal O^{A_1,A_2}_{\mu_1,\mu_2,}\{f\}(w)\|_q\le (2\pi)^{\frac{1}{q}-\frac{1}{p}} |b_1b_2|^{\frac{1}{q}-\frac{1}{2}}\|f(x)\|_p.
\end{equation}
\end{lemma}

\begin{theorem}[Hausdorff-Young inequality for $\mathbb O-$OLCT]\label{th hdrfoolct} For $1\le p\le 2$ and $\frac{1}{p}+\frac{1}{q}=1,$ we have
\begin{equation}\label{eqnhdrof oolct}
\|\mathcal O^{A_1,A_2,A_3}_{\mu_1,\mu_2,\mu_4}\{f\}(w)\|_q\le {(2\pi)^{\frac{1}{2q}-\frac{1}{p}}} |b_1b_2|^{\frac{1}{q}-\frac{1}{2}}|b_3|^{-\frac{1}{2q}}\|f(x)\|_p.
\end{equation}
\end{theorem}
\begin{proof}
From (\ref{qolct}), we obtain
\begin{eqnarray*}
\|\mathcal O^{A_1,A_2,A_3}_{\mu_1,\mu_2,\mu_4}\{f\}(w)\|_q&=&\frac{1}{(2\pi|b_3|)^{\frac{1}{2q}}}\left(\|\mathcal O^{A_1,A_2}_{\mu_1,\mu_2}\{g_e\}(w)\|_q+\mathcal O^{A_1,A_2}_{\mu_1,\mu_2}\{g_o\}(w)\|_q\right.\\
&&+\left.\mathcal O^{A_1,A_2}_{\mu_1,\mu_2}\{h_e\}(w)\|_q+\mathcal O^{A_1,A_2}_{\mu_1,\mu_2}\{h_o\}(w)\|_q\right)
\end{eqnarray*}
Now applying lemma \ref{lem hdrfolct} to the R.H.S of above equation, we obtain
\begin{eqnarray*}
\|\mathcal O^{A_1,A_2,A_3}_{\mu_1,\mu_2,\mu_4}\{f\}(w)\|_q&\le&\frac{1}{(2\pi|b_3|)^{\frac{1}{2q}}}(2\pi)^{\frac{1}{q}-\frac{1}{p}} |b_1b_2|^{\frac{1}{q}-\frac{1}{2}}\left(\|g_e(x)\|_p+\|g_o(x)\|_p\right.\\
&&\left.+\|h_e(x)\|_p+\|h_o(x)\|_p\right).\\
&=&{(2\pi)^{\frac{1}{2q}-\frac{1}{p}}} |b_1b_2|^{\frac{1}{q}-\frac{1}{2}}|b_3|^{-\frac{1}{2q}}\|f(x)\|_p\\
\end{eqnarray*}
where last equality follows from (\ref{evenoddfun}).\\
Which completes the proof
\end{proof}


\subsection{Local uncertainty principle for  $\mathbb O-$OLCT}\ \\
 Local uncertainty principle states that if $f$ is highly localized, then the octonion Fourier transform can not
 be concentrated in a small neighborhood of two or more separated points. We shall establish Local uncertainty principle for  $\mathbb O-$OLCT in this subsection.

\begin{lemma}[Local uncertainty principle for QOLCT]\label{lem local qolct}\cite{ucp 2sd qolct}\ \\
 (1) For $0<\alpha<1$ and for all $f\in L^2(\mathbb R^2,\mathbb H),$  there is a constant $M_\alpha$ and all measurable set $E\subset \mathbb R^3$ that holds
\begin{equation}\label{eqn local qolct}
\int_{bE}|\mathcal O^{\mu_1,\mu_2}_{A_1,A_2}\{f\}(w)|^2dw\le M_\alpha|E|^\alpha\||x|^\alpha f\|^2_2.
\end{equation}
(2) If $\alpha >1$,and for all $f\in L^2(\mathbb R^2,\mathbb H),$  there is a constant $M_\alpha$ and all measurable set $E\subset \mathbb R^3$ that holds
\begin{equation}\label{eqn local qolct 1}
\int_{bE}|\mathcal O^{\mu_1,\mu_2}_{A_1,A_2}\{f\}(w)|^2dw\le M_\alpha|b_1b_2|^{\alpha-\frac{1}{\alpha}}|E|^\alpha\|f\|^{2-2\alpha}_2\||x|^\alpha f\|^{\frac{2}{\alpha}}_2,
\end{equation}
\begin{equation}
 M_\alpha=\left\{\begin{array}{cc} \dfrac{(1+\alpha^2)}{\alpha^{2\alpha}}(2-2\alpha)^{\alpha-2}, &0<\alpha <1,\\\\ \dfrac{\pi}{\alpha \Gamma(1/2)}\Gamma\left(\frac{1}{\alpha}\right)\Gamma\left(1-\frac{1}{\alpha}\right)(\alpha-1)^\alpha\left(1-\frac{1}{\alpha}\right)^{-1}, & \alpha >1 ,\\\\
\end{array}\right.
\end{equation}
\end{lemma}

\begin{theorem}[Local uncertainty principle for  $\mathbb O-$OLCT]\label{th local oolct}\ \\
 (1) For $0<\alpha<1$ and for all $f\in L^2(\mathbb R^3,\mathbb O),$  there is a constant $M_\alpha$ and all measurable set $E\subset \mathbb R^3$ that holds
\begin{equation}\label{eqn local qolct}
\int_{bE}|\mathcal O^{\mu_1,\mu_2}_{A_1,A_2}\{f\}(w)|^2dw\le \frac{1}{2\pi|b_3|}M_\alpha|E|^\alpha\||x|^\alpha f\|^2_2.
\end{equation}
(2) If $\alpha >1$,and for all $f\in L^2(\mathbb R^3,\mathbb O),$  there is a constant $M_\alpha$ and all measurable set $E\subset \mathbb R^3$ that holds
\begin{equation}\label{eqn local qolct 1}
\int_{bE}|\mathcal O^{\mu_1,\mu_2}_{A_1,A_2}\{f\}(w)|^2dw\le \frac{1}{2\pi|b_3|} M_\alpha|b_1b_2|^{\alpha-\frac{1}{\alpha}}|E|^\alpha\|f\|^{2-2\alpha}_2\||x|^\alpha f\|^{\frac{2}{\alpha}}_2,
\end{equation}
\begin{equation}
 M_\alpha=\left\{\begin{array}{cc} \dfrac{(1+\alpha^2)}{\alpha^{2\alpha}}(2-2\alpha)^{\alpha-2}, &0<\alpha <1,\\\\ \dfrac{\pi}{\alpha \Gamma(1/2)}\Gamma\left(\frac{1}{\alpha}\right)\Gamma\left(1-\frac{1}{\alpha}\right)(\alpha-1)^\alpha\left(1-\frac{1}{\alpha}\right)^{-1}, & \alpha >1 ,\\\\
\end{array}\right.
\end{equation}
\end{theorem}
\begin{proof}
 By the splitting of  $\mathbb O-$OLCT in (\ref{qolct}),we have
\begin{eqnarray*}
\nonumber\mathcal O^{A_1,A_2,A_3}_{\mu_1,\mu_2,\mu_4}\{f\}(w)&=&\frac{1}{\sqrt{2\pi|b_3|}}\int_{\mathbb R^3}g_e(x)K_{A_1}^{\mu_1}(x_1,w_1)K_{A_2}^{\mu_2}(x_2,w_2)c_3dx\\
\nonumber&&+\frac{1}{\sqrt{2\pi|b_3|}}\int_{\mathbb R^3}h_o(x)K_{A_1}^{-\mu_1}(x_1,w_1)K_{A_2}^{-\mu_2}(x_2,w_2)s_3dx\\
\nonumber&&+\left(\frac{1}{\sqrt{2\pi|b_3|}}\int_{\mathbb R^3}h_e(x)K_{A_1}^{-\mu_1}(x_1,w_1)K_{A_2}^{-\mu_2}(x_2,w_2)c_3dx\right.\\
\nonumber&&-\left.\frac{1}{\sqrt{2\pi|b_3|}}\int_{\mathbb R^3}g_o(x)K_{A_1}^{\mu_1}(x_1,w_1)K_{A_2}^{\mu_2}(x_2,w_2)s_3dx\right)\mu_4.\\
\end{eqnarray*}
Setting $f_m(x)=g_e(x_{+++})+h_o(x_{+-+})e_2$ and $f_n(x)x=h_e(x_{+++})-g_o(x_{+-+})e_2$. Then $\mathbb O-$OLCT can be written as the combination of two QOLCTs as
\begin{eqnarray}\label{oolct fm fn}
\nonumber\|\mathcal O^{A_1,A_2,A_3}_{\mu_1,\mu_2,\mu_4}\{f\}(w)\|^2&=&\frac{1}{2\pi|b_3|}\|\mathcal O^{A_1,A_2}_{\mu_1,\mu_2}\{f_m\}(w)\|^2+\|\mathcal O^{A_1,A_2}_{\mu_1,\mu_2}\{f_n\}(w)\|^2\\
\end{eqnarray}
Now, by lemma \ref{lem local qolct} we have
\begin{equation}\label{www}
\int_{bE}|\mathcal O^{\mu_1,\mu_2}_{A_1,A_2}\{f_j\}(w)|^2dw\le M_\alpha|E|^\alpha\||x|^\alpha f_j\|^2_2,\quad j=m,n
\end{equation}
Therefore from (\ref{oolct fm fn}),(\ref{www}) and using lemma \ref{lem1},we have
\begin{eqnarray*}
\int_{bE}\|O^{A_1,A_2,A_3}_{\mu_1,\mu_2,\mu_4}\{f\}(w)\|^2dw&=&\frac{1}{2\pi|b_3|}\int_{bE}\left(\|\mathcal O^{A_1,A_2}_{\mu_1,\mu_2}\{f_m\}(w)\|^2dw+\|\mathcal O^{A_1,A_2}_{\mu_1,\mu_2}\{f_n\}(w)\|^2dw\right)\\
&\le&\frac{1}{2\pi|b_3|}M_\alpha|E|^\alpha\left(\||x|^\alpha f_m\|^2+\||x|^\alpha f_n\|^2_2 \right)\\
&=&\frac{1}{2\pi|b_3|}M_\alpha|E|^\alpha\||x|^\alpha f\|^2_2 \\
\end{eqnarray*}
which completes the proof.\\
Similarly we can easily prove (\ref{eqn local qolct 1}).
\end{proof}

\section{Conclusions}\ \\

In this paper, based on the association between the $\mathbb O-$OLCT the QOLCT via split norm, we have  established some basic properties of the proposed transform including the inversion formula and energy conservation. These
results are very important for their applications in digital signal and image processing. Finally,  the uncertainty
 inequalities for the  $\mathbb O-$OLCT such as logarithmic uncertainty inequality,Hausdorff-Young inequality and  local uncertainty
 are obtained.
In
our future works, we will discuss the physical significance and engineering background of this paper. Moreover, we will formulate convolution and correlation theorems for the $\mathbb O-$OLCT.

\end{document}